\documentclass[10pt,conference]{IEEEtran}
\IEEEoverridecommandlockouts

\usepackage{epsfig,rotating,setspace,latexsym,amsmath,epsf,amssymb,amsfonts,bm,cite,algorithm,graphicx,epsf,authblk,epstopdf,color,algpseudocode,bbm}
\usepackage{amsthm}

\usepackage{algorithm} 
\usepackage{algpseudocode} 
\usepackage{booktabs}
\usepackage{subfigure}
\usepackage{mathtools,blkarray}
\usepackage{caption}
\usepackage{nicematrix}
\usepackage{multirow} % Required for multirows
\usepackage{lscape} 
\usepackage{subcaption}

\newcounter{procedure}%        counter for protocols
\newcounter{algorithm saved}% the real counter for algorithms

\newtheorem{remark}{Remark}

\newtheorem{corollary}{Corollary}

\newtheorem{theorem}{Theorem}
\newtheorem{lemma}{Lemma}

\usepackage{bm} % Best package for bold math

\makeatletter


\allowdisplaybreaks

\normalsize

\ifCLASSINFOpdf
\else
\fi

\begin{document}

\title{Eigenspace-Based Clustering for Personalized System Identification}
\author{Abdulmoneam Ali$~~~~~~\qquad$Dipankar Maity$~~~~~~\qquad$Ahmed Arafa\\Department of Electrical and Computer Engineering\\ University of North Carolina at Charlotte, NC 28223\\
$\quad$\emph{aali28@charlotte.edu}$\qquad$\emph{dmaity@charlotte.edu}$\qquad$\emph{aarafa@charlotte.edu}
\thanks{This work of A. Ali and A. Arafa was supported by the U.S. National Science Foundation under Grant ECCS 21-46099. The work of D. Maity was supported by the U.S. National Science Foundation CAREER award under Grant 2443349.}}

\markboth{IEEE Transactions on Communications}%
{Submitted paper}

\maketitle

\begin{abstract}
We study the problem of system identification in heterogeneous settings, where different systems may follow distinct underlying dynamics. Existing clustered system identification approaches often rely on iterative training-based cluster assignment, which can be sensitive to learning uncertainty and model initialization. In contrast, we propose a \emph{one-shot, training-free} clustering method that identifies similar systems using the structure of their locally observed data. Specifically, each system estimates a local state covariance matrix, and cluster identities are inferred by measuring the alignment between the leading covariance eigenspaces of different systems. We provide a mathematical interpretation of the proposed similarity score and develop a finite-sample analysis that characterizes how covariance estimation error induces eigenspace perturbations in terms of the underlying system dynamics. We then bound the probability of pairwise false eigenspace alignment and establish sufficient conditions for uniform inter-cluster eigenspace separation with high probability. Numerical experiments demonstrate that the proposed eigenspace-based clustering method effectively identifies systems with shared dynamics, leading to lower personalized model-estimation error compared with training-based clustering and non-clustered baselines.  
\end{abstract}

%keywords
% \IEEEpeerreviewmaketitle

%===
\section{Introduction}
Dynamical systems whose evolution is governed by unknown parameters serve as a central abstraction in control theory, robotics, and learning-based decision making~\cite{gossip}. The task of estimating these latent parameters from empirical observations---commonly referred to as system identification---is a fundamental prerequisite for effective model-based control, time-series forecasting, and reinforcement learning. 

In many practical settings, the data required for robust identification are naturally distributed across multiple distinct systems, which raises significant privacy concerns about sharing raw trajectories with a centralized entity. To address this bottleneck, federated learning (FL) has emerged as a privacy-preserving framework that enables different systems to learn collaboratively while keeping the data decentralized~\cite{pmlr-v54-mcmahan17a}.

Despite these advantages, the practical deployment of FL faces several challenges, including communication bottlenecks, discrepancies among systems in communication and computation capabilities, and data heterogeneity~\cite{OTA}. Data heterogeneity means that the underlying statistical distributions of local data vary significantly across participating systems. This non-i.i.d. (independent and identically distributed) nature makes it difficult to learn a single model that fits all local datasets. To address this challenge, personalized federated learning (PFL) has been proposed as a promising approach that shifts the paradigm from learning a single global model to learning models tailored to individual systems.

A rich body of literature has explored various strategies for enabling different systems to contribute to the learning of personalized models. These approaches range from meta-learning and aggregation-based strategies to clustering similar systems. A central challenge in clustering-based personalization is identifying which systems are truly similar. Most existing approaches address this problem by analyzing the training trajectories of each system with respect to different personalized models. In particular, systems are grouped according to the personalized model that minimizes their local loss function. 

\textbf{Related Work}. 
The work in~\cite{pmlr-FedSysID} demonstrates that  collaboration enables systems to achieve improved finite-time convergence guarantees, provided their dynamics are similar. This similarity is characterized by the maximum distance between the corresponding state-transition and control-input matrices.  Conversely, for divergent systems, collaboration can be detrimental to the learning process. 

To mitigate this system heterogeneity issue, the authors in~\cite{toso2023learning} propose an iterative model-based clustering technique based on an alternating approach. Specifically, each system first identifies the model cluster that minimizes its local loss, then performs a local gradient update and shares both the update and the estimated cluster identity with a central server. The server aggregates all local gradients belonging to the same cluster, repeating this process until both the cluster assignments and the models converge. However, a significant drawback of this iterative process is its strict requirement for good initialization (i.e., initial cluster models are assumed to be within a close neighborhood of the true dynamics), which is an optimistic assumption that may not hold in practice. Rather than prescribing a fixed number of clusters, the authors in~\cite{kecceci2024novel} propose an incremental clustering methodology. To mitigate the training instability induced by fluctuating cluster assignments across training rounds, a regularization term is incorporated into the loss function. Nevertheless, a common limitation shared by these existing works is their reliance on iterative, training-dependent architectures that remain inherently sensitive to the initialization.

%To fill these gaps, 
To overcome these challenges, we propose a novel eigenspace-based PFL framework for system identification. In contrast to training-based approaches, the proposed method infers cluster identities by extracting structural information from locally observed trajectories through the leading eigenspaces of the empirical state covariance matrices. This enables the clustering phase to be performed once, prior to collaborative model training, thereby decoupling cluster identity estimation from parameter estimation. Our main contributions are threefold, summarized as follows:

\begin{itemize}
    \item \textbf{One-shot, training-free clustering methodology:} Unlike iterative alternating-minimization techniques that require model training during cluster assignment and can be sensitive to model initialization, we propose a one-shot, training-free clustering approach executed before collaborative training. Each system computes the leading eigenspace of its local state covariance matrix, and pairwise similarities are formed by measuring eigenspace alignments across systems. This avoids dependence on model initialization and reduces the communication required for cluster identity estimation.

    \item \textbf{Theoretical performance guarantees:} We provide a mathematical interpretation of the proposed eigenspace-based similarity score and develop a finite-sample analysis of eigenspace estimation and separation. Specifically, we first bound the covariance estimation error and characterize the resulting eigenspace perturbation. We then use these results to derive a pairwise false eigenspace alignment probability bound and establish sufficient conditions for uniform inter-cluster eigenspace separation with high probability. The resulting bounds reveal how the number of trajectories, eigengaps, inter-cluster separation, and system-dependent covariance structure affect the reliability of the estimated eigenspace representations. 
 
    \item \textbf{Empirical validation:} We validate the proposed method on heterogeneous LTI systems and compare it with an iterative training-based clustering baseline and a non-clustered global model. The results show that the proposed one-shot eigenspace-based clustering method identifies systems with shared dynamics, improves personalized model-estimation accuracy, and achieves decreasing clustering misclassification rates as the number of trajectories increases, consistent with the finite-sample theory. 
\end{itemize}

%\noindent \textbf{Paper Organization.}

\noindent \textbf{Notation and Organization.} We use lowercase letters for scalars, bold lowercase symbols for vectors, and bold uppercase symbols for matrices throughout this paper. For a matrix $\mathbf{A}$, $\mathbf{A}[i,j]$ denotes the element in row $i$ and column $j$, and $\mathbf{A}[i]$ denotes its $i$th row. We represent the set $\{1,2,\dots,M\}$ by $[M]$. 
For matrices, $\|\mathbf{A}\|_2$ and $\|\mathbf{A}\|_F$ denote the spectral and Frobenius norms, respectively, while for a vector $\mathbf{a}$, $\|\mathbf{a}\|_2$ denotes its Euclidean norm. An identity matrix of appropriate dimension is denoted by $\mathbf{I}$. The remainder of the paper is organized as follows. Section~\ref{sec:sys_model} presents the system model. Section~\ref{sec:data_sim_clustering} introduces the proposed eigenspace-based clustering algorithm, while Section~\ref{sec:theor_guarantee} presents its theoretical guarantees. Section~\ref{sec:experiments} presents the experimental setup and results. Section~\ref{sec:conclusion} concludes the paper.

%===
\section{System Model}\label{sec:sys_model}

We consider a central server and a set of $M$ linear time-invariant (LTI) systems characterized by the following state-space dynamics: 
\begin{align} \label{eq:state-space}
    \mathbf{x}_{t+1}^{(i)} = \mathbf{A}^{(i)} \mathbf{x}_{t}^{(i)} + \mathbf{B}^{(i)} \mathbf{u}_t^{(i)} + \mathbf{w}_{t}^{(i)}, \, t = 0,1, \dots, T-1,
\end{align}
where $\mathbf{x}_t^{(i)} \in \mathbb{R}^{n_{x}}$ denotes the state, $\mathbf{u}_t^{(i)} \in \mathbb{R}^{n_{u}}$ denotes the control input, and $\mathbf{w}_t^{(i)} \in \mathbb{R}^{n_{x}}$ denotes the process noise for each system $i \in [M]$ at time $t$. The matrices $\mathbf{A}^{(i)} \in \mathbb{R}^{n_x \times n_x}$ and $\mathbf{B}^{(i)} \in \mathbb{R}^{n_x \times n_u}$ represent the state-transition and control-input dynamics, respectively. Following~\cite{toso2023learning, kecceci2024novel, kecceci2025redefining}, for each system $i \in [M]$, the initial state $\mathbf{x}_{0}^{(i)}$, the control inputs $\mathbf{u}_t^{(i)}$, and noise terms $\mathbf{w}_t^{(i)}$ for $t=0, \dots,T-1$ are mutually independent zero-mean Gaussian random vectors with covariance matrices $\sigma^2_{x,i} \mathbf{I}_{n_{x}}, \sigma^2_{u,i} \mathbf{I}_{n_{u}}$, and $\sigma^2_{w,i} \mathbf{I}_{n_{x}}$, respectively.

We assume that we have $M$ trajectory datasets, each generated by one of $K$ underlying linear systems, where $K \ll M$. These $K$ distinct linear systems constitute clusters, denoted by $\mathcal{C}_1,\dots,\mathcal{C}_K$. For any system $i \in \mathcal{C}_k$, the dynamics are governed by the matrices $(\mathbf{A}_k, \mathbf{B}_k)$, i.e., $\mathbf{A}^{(i)} = \mathbf{A}_k$ and $\mathbf{B}^{(i)} = \mathbf{B}_k$. Due to stochastic noise, the realized trajectories for each system, even within the same cluster, are different.  

For each system $i \in \mathcal{C}_k$, we define a rollout as a single trajectory of state-input pairs $\{\mathbf{x}_t^{(i)}, \mathbf{u}_t^{(i)} \}_{t=0}^{T-1}$. We assume that $N_i$ such rollouts, each of length $T$, are collected and denoted by $\{\mathbf{x}_{l,t}^{(i)}, \mathbf{u}_{l,t}^{(i)} \}_{t=0}^{T-1}$ for $l=1,\dots,N_i$. Here, $l$ and $t$ index the rollout and time step, respectively. 
Let $\mathbf{z}_{l,t}^{(i)}
:= \big[\mathbf{x}_{l,t}^{(i)\top},\mathbf{u}_{l,t}^{(i)\top}\big]^{\top}
\in\mathbb{R}^{d}$, where $d:=n_x+n_u$.
The dynamics of any system $i \in \mathcal{C}_k$ within cluster $k \in [K]$ can be compactly expressed as
\begin{align}
    \mathbf{x}_{l,t+1}^{(i)} = \mathbf{\Theta}_{k}\, \mathbf{z}_{l,t}^{(i)} + \mathbf{w}_{l,t}^{(i)}, 
    \quad \forall 1 \leq l \leq N_i,\; 0 \leq t \leq T-1,
\end{align}
where $\mathbf{\Theta}_k = [\mathbf{A}_k \; \mathbf{B}_k] \in \mathbb{R}^{n_x \times d}$ 
is the concatenated ground-truth system matrix for cluster $k$.

To facilitate batch processing, we organize the data from all $N_i$ rollouts of each system $i\in [M]$ into comprehensive matrices. In particular, for a single rollout $l$, the data are concatenated according to $\mathbf{X}_{l}^{(i)} = [\mathbf{x}_{l, T}^{(i)} \dots \mathbf{x}_{l, 1}^{(i)}] \in \mathbb{R}^{n_x \times T}$, $\mathbf{Z}_{l}^{(i)} = [\mathbf{z}_{l, T-1}^{(i)} \dots \mathbf{z}_{l, 0}^{(i)}] \in \mathbb{R}^{d \times T}$, and $\mathbf{W}_{l}^{(i)} = [\mathbf{w}_{l, T-1}^{(i)} \dots \mathbf{w}_{l, 0}^{(i)}] \in \mathbb{R}^{n_x \times T}$.\footnote{We adopt a reverse-time stacking convention, consistent with \cite{toso2023learning}, 
where state vectors are ordered from $t=T$ to $t=1$ and inputs from $t=T-1$ to $t=0$. 
This choice is purely notational and does not affect the system dynamics or estimation results.
} These matrices are then concatenated into a single batch representation $\mathbf{X}^{(i)} = [\mathbf{X}_{1}^{(i)} \dots \mathbf{X}_{N_i}^{(i)}] \in \mathbb{R}^{n_x \times N_iT}$, $\mathbf{Z}^{(i)} = [\mathbf{Z}_{1}^{(i)} \dots \mathbf{Z}_{N_i}^{(i)}] \in \mathbb{R}^{d \times N_iT}$, and $\mathbf{W}^{(i)} = [\mathbf{W}_{1}^{(i)} \dots \mathbf{W}_{N_i}^{(i)}] \in \mathbb{R}^{n_x \times N_iT}$. Thus, for each system $i \in \mathcal{C}_k, \, \forall k \in [K]$, the system dynamics can be written as
\begin{align}
    \mathbf{X}^{(i)} = \mathbf{\Theta}_k\, \mathbf{Z}^{(i)} + \mathbf{W}^{(i)}.
\end{align}

Given this batch representation, each system $i \in [M]$ has access only to its local data $(\mathbf{X}^{(i)}, \mathbf{Z}^{(i)})$ and locally minimizes the least-squares identification objective function \begin{equation}
    J^{(i)}(\mathbf{\Theta}) = \left\|\mathbf{X}^{(i)}-\mathbf{\Theta}\, \mathbf{Z}^{(i)}\right\|_F^2.
\end{equation}

However, while systems belonging to the same cluster share a common ground-truth model $\mathbf{\Theta}_k$, each system remains unaware of its cluster identity. This creates a coupled learning problem: collaborative parameter estimation requires grouping systems with shared dynamics, yet training-based clustering methods typically rely on intermediate model estimates to determine such groups. To break this coupling, we propose a clustering mechanism that identifies systems with shared dynamics before collaborative training, so that parameter estimation can then be performed within each estimated cluster.

Next, we discuss our proposed clustering algorithm.

%===
\section{Eigenspace-Based System Identification Clustering} \label{sec:data_sim_clustering}

To enable systems with the same model dynamics to learn cooperatively, such systems should be assigned the same cluster identity. Systems assigned to cluster $k$ jointly learn the model parameter $\mathbf{\Theta}_k$. Specifically, during global communication round $h$, each system $i$ assigned to cluster $k$ computes the local matrix-valued gradient of its least-squares objective, $\nabla J_i^{(h)}:=\nabla_{\mathbf{\Theta}} J^{(i)}\!\left(\widehat{\mathbf{\Theta}}_k^{(h)}\right)$, as given in Step~\ref{eq:GD} of Algorithm~\ref{alg:main_algo}, and sends the local gradient to the server. The server aggregates the received gradients from all systems assigned to cluster $k$ and updates the cluster-level estimate $\widehat{\mathbf{\Theta}}_k^{(h+1)}$, which is then broadcast back to the systems in that cluster. The main challenge in applying this PFL procedure is that the cluster identities---that is, which systems share the same model dynamics---are unknown. 

To efficiently estimate the cluster identities of the systems, we build on a data-similarity methodology originally proposed in~\cite{data_valuation} and later extended to hierarchical federated multi-task learning~\cite{ds_asilomar24} and noisy-label detection and correction~\cite{RCC_icc25, ali2026fbnll}.
The central idea is to compare datasets through the structure of their leading covariance eigenspaces, which capture how data variation is distributed across principal directions. 
In this work, we adapt this idea to personalized system identification by constructing an eigenspace-based similarity score from local state trajectories. Unlike prior applications, we further develop a system-identification-specific theoretical framework that interprets the similarity score, characterizes finite-sample eigenspace perturbations induced by covariance estimation, and provides high probability inter-cluster eigenspace-separation guarantees in terms of the underlying system dynamics. 

For mathematical tractability, we form an i.i.d. state matrix using the final state vectors sampled across independent rollouts. In particular, each system $i$ constructs 
\[
\mathbf{X}_{T}^{(i)} := 
\big[\mathbf{x}_{1,T}^{(i)} \;\; \mathbf{x}_{2,T}^{(i)} \;\; \cdots \;\; \mathbf{x}_{N_i,T}^{(i)}\big]
\in \mathbb{R}^{n_x \times N_i}.\]

Since the time index is fixed throughout the clustering stage, we denote $\mathbf{X}_c^{(i)}:= \mathbf{X}_{T}^{(i)}$ for notational simplicity.
Note that $\mathbf{X}_c^{(i)}$ is used solely for clustering and does not affect the system identification updates, which rely on the full dataset $\mathbf{X}^{(i)}$.

% they are centralized
The proposed cluster identity estimation procedure consists of three main steps. First, each system $i$ computes the eigendecomposition of its local sample covariance matrix,
\[
\widehat{\mathbf{\Sigma}}^{(i)} = \frac{1}{N_i} \mathbf{X}_c^{(i)}\, \mathbf{X}_c^{(i)\top},
\]
as
 \begin{align} \label{eq:est_eigendecomp}
     \widehat{\mathbf{\Sigma}}^{(i)} = \widehat{\bm{V}}^{(i)} \widehat{\mathbf{\Lambda}}^{(i)} \widehat{\bm{V}}^{(i) \top},
 \end{align}
where $\widehat{\mathbf{\Lambda}}^{(i)}$ contains the eigenvalues in decreasing order, and $\widehat{\bm{V}}^{(i)}$ contains the corresponding eigenvectors.

In the second step, systems exchange eigenvectors through the server to measure the alignment of their covariance eigenspaces with directions learned by other systems.\footnote{This exchange avoids sharing raw trajectories or sample covariance matrices.} 
This is carried out by applying the local sample covariance matrix to eigenvector directions received from other systems and evaluating the Euclidean norm:
\begin{align}\label{eq:v_image}
    a_{q}^{(i,\, j)}=\left\| \widehat{\mathbf{\Sigma}}^{(i)}\widehat{\mathbf{v}}_{q}^{(j)}\right\|_2, \quad \forall q \in [r]. 
\end{align}

Here, $\widehat{\mathbf v}_q^{(j)}$ is the $q$th leading eigenvector
associated with the $q$th largest eigenvalue of system $j$,
and $r \leq n_x$ denotes the number of leading eigenvectors utilized.
Based on the local eigenvalues and the directional covariance response, each system quantifies the normalized directional similarity 
\begin{align}
   \rho_{q}^{(i,\, j)}=\frac{\min\{\widehat{\lambda}_q^{(i)},a_q^{(i,\,j)}\}} {\max\{\widehat{\lambda}_q^{(i)},a_q^{(i,\,j)}\}}, \quad \forall q \in [r], \label{eq:lambd(i,j)}
\end{align}
where $\widehat{\lambda}^{(i)}_q$ is the $q$th largest eigenvalue of $\widehat{\mathbf{\Sigma}}^{(i)}$, assumed positive for all $i \in [M]$ and $q \in [r]$. The $\min/\max$ normalization ensures that $\rho_{q}^{(i,\, j)} \in [0,1]$. In particular, if $\widehat{\mathbf{v}}_{q}^{(j)} = \pm  \widehat{\mathbf{v}}_{q}^{(i)}$, then $a_{q}^{(i,\, j)} = \widehat{\lambda}_{q}^{(i)}$, which implies that $\rho_q^{(i,j)} =1$. Otherwise, $a_{q}^{(i,\, j)}$ reflects how strongly the received eigenvector aligns with the local covariance structure.

Given the $r$ values $\{\rho_{q}^{(i,j)}\}_{q=1}^r$, system $i$ computes a scalar similarity score, $s_{ij}$, using the geometric mean:
\begin{align}
    s_{ij}=\left(\prod_{q=1}^{r} \rho_{q}^{(i,\,j)}\right)^{\frac{1}{r}}.\label{eq:s(i,j)}
\end{align}

Thus, system $i$ assigns a high similarity score to system $j$ only when the received leading eigenvectors exhibit consistent alignment with its local covariance structure across the $r$ dominant directions. 

Finally, systems share similarity scores obtained from \eqref{eq:s(i,j)} with the server, which constructs a symmetric pairwise similarity matrix as follows:
\begin{align}\label{eq:avg_s}
   \mathbf{S}[i,j]=\frac{s_{ij} + s_{ji}}{2}, \quad \forall i, j \in [M].
\end{align}
By construction, $\mathbf{S}\in\mathbb{R}^{M\times M}$ is symmetric. The server then runs the spectral clustering algorithm~\cite{pml1Book} using $\mathbf{S}$ as the affinity matrix to infer the cluster assignment of each system. We encode the cluster assignments in a matrix $\mathbf{CI} \in \{0,1\}^{M\times K}$, where $\mathbf{CI}[i,k]=1$ indicates that system $i$ is assigned to cluster $k$. Each system is assigned to exactly one cluster.
The detailed steps are shown in Algorithm~\ref{alg:main_algo}.

To illustrate a realization of $\mathbf{S}$, Table~\ref{tab:example_of_S} shows an example with five systems and three clusters. Systems 1, 2, and 3 share the same underlying dynamics $\boldsymbol{\Theta}$, while systems 4 and 5 each belong to a distinct cluster. The resulting matrix exhibits larger similarity values among systems 1--3 and smaller similarities between systems from different clusters.

\begin{table}[h]
\centering
\caption{\small Illustrative similarity matrix $\mathbf{S}$.}
\label{tab:example_of_S}
\begin{tabular}{lccccc}
\toprule
System & Sys.~1 & Sys.~2 & Sys.~3 & Sys.~4 & Sys.~5 \\
\midrule
Sys.~1 & 1.00 & 0.94 & 0.93 & 0.50 & 0.64\\
Sys.~2 & 0.94 & 1.00 & 0.94 & 0.52 & 0.58\\
Sys.~3 & 0.93 & 0.94 & 1.00 & 0.55 & 0.62\\
Sys.~4 & 0.50 & 0.52 & 0.55 & 1.00 & 0.71\\
Sys.~5 & 0.64 & 0.58 & 0.62 & 0.71 & 1.00 \\
\bottomrule
\end{tabular}
\end{table}
% % Based on X

%================================
\begin{algorithm}[t]
    \caption{\small Eigenspace-Based Clustered System Identification}\label{alg:main_algo}
\begin{algorithmic}[1]
\small 
\State \textbf{Input}: number of systems $M$, number of clusters $K$, number of global iterations $R$, step size $\beta$, number of leading eigenvectors~$r$
\State \textbf{Output}: cluster assignment matrix $\mathbf{CI}$ and cluster-level model estimates $\{\widehat{\mathbf{\Theta}}_k^{(R)}\}_{k=1}^K$
\Statex \underline{\textit{(Cluster Identity Estimation)}}
\State Initialize the similarity matrix $\mathbf{S}  = \mathbf{I}_M$

\For{$i \in [M]$} 
    \State $\widehat{\mathbf{\Sigma}}^{(i)} \gets \frac{1}{N_i}\mathbf{X}_c^{(i)}\mathbf{X}_c^{(i)\top}$
\State Compute eigendecomposition $\widehat{\mathbf{\Sigma}}^{(i)} = \widehat{\mathbf{V}}^{(i)} \widehat{\mathbf{\Lambda}}^{(i)}\widehat{\mathbf V}^{(i)\top}$ 
\State Upload top-$r$ eigenvectors $\widehat{\mathbf{V}}_{0}^{(i)} \in \mathbb{R}^{n_x \times r}$ to the server
\EndFor
\State The server broadcasts $\{\widehat{\mathbf{V}}_{0}^{(j)}\}_{j=1}^M$ to all systems

\For {$i \in [M]$}
      \For {$ j  \in [M],\,  j\neq i$} 
      \State Compute $s_{ij}$ using \eqref{eq:v_image}--\eqref{eq:s(i,j)}; send $s_{ij}$ to the server
\State Server stores $\mathbf{S}[i,j] \gets s_{ij}$
    \EndFor
\EndFor
\State Symmetrize $\mathbf{S}$ using \eqref{eq:avg_s}; run spectral clustering to obtain $\mathbf{CI} \in \{0,1\}^{M \times K}$
\Statex \underline{\textit{Clustered System Identification}}
\State Server broadcasts $\widehat{\mathbf{\Theta}}_k^{(0)}$ for all $k \in [K]$
\For {$h = 0, \dots, R-1$}
    \For{ $i \in [M]$}
        \State Find assigned cluster $k$ such that $\mathbf{CI}[i,k] = 1$
        \Statex \hspace{0.45in}Compute local gradient matrix $\nabla J_{i}^{(h)}$:
        \State $\nabla J_{i}^{(h)} = - 2\left(\mathbf{X}^{(i)} - \widehat{\mathbf{\Theta}}_{k}^{(h)} \mathbf{Z}^{(i)}\right) \, \mathbf{Z}^{(i),{\top}}$ \label{eq:GD}
        \State System $i$ sends $\nabla J_{i}^{(h)}$ to the server
    \EndFor
    
    \Statex \underline{\textit{(Server Aggregation)}} 
    \For{each cluster $k \in [K]$}
        \State $|\widehat{\mathcal{C}}_k| = \sum_{i=1}^M \mathbf{CI}[i,k]$ 
        \State $\widehat{\mathbf{\Theta}}_k^{(h+1)} = \widehat{\mathbf{\Theta}}_k^{(h)} - \frac{\beta}{|\widehat{\mathcal{C}}_k|} \sum_{i=1}^M \mathbf{CI}[i,k] \nabla J_{i}^{(h)}$
    \EndFor
\EndFor
\end{algorithmic}
\end{algorithm}

%=====================================
\section{Theoretical Analysis} \label{sec:theor_guarantee}

\subsection{Interpretation of the similarity score}
We denote the population covariance matrix of the final state of system $i$ by $\mathbf{\Sigma}^{(i)} = \mathbb{E}[\mathbf{x}_T^{(i)} \mathbf{x}_T^{(i) \top}] \in \mathbb{R}^{n_x \times n_x}$.
Since $\mathbf{\Sigma}^{(i)}$ is symmetric and positive semidefinite, it admits the eigendecomposition
\begin{equation}\label{eq:true_eigendecomp}
 \mathbf{\Sigma}^{(i)} = \mathbf{V}^{(i)} \mathbf{\Lambda}^{(i)} \mathbf{V}^{(i)\top}. 
\end{equation}

We now characterize the population-level interpretation of a similarity score equal to one. 
 \begin{lemma} \label{lemma:sim_eigens} 
 Let $\mathbf{\Sigma}^{(i)}, \mathbf{\Sigma}^{(j)} \in \mathbb{R}^{n_x \times n_x}$, with $i \neq j$, be symmetric positive semidefinite 
 matrices with distinct eigenvalues $\lambda_1^{(i)} > \lambda_2^{(i)} > \dots >\lambda_{n_x}^{(i)}$ and $\lambda_1^{(j)} > \lambda_2^{(j)} > \dots >\lambda_{n_x}^{(j)}$. Let $\mathbf{V}^{(i)}=[\mathbf{v}_1^{(i)} \, \mathbf{v}_2^{(i)} \dots \mathbf{v}_{n_x}^{(i)} ]$ and $\mathbf{V}^{(j)}=[\mathbf{v}_1^{(j)} \, \mathbf{v}_2^{(j)} \dots \mathbf{v}_{n_x}^{(j)} ]$ be the corresponding orthonormal eigenvector matrices ordered according to decreasing eigenvalues. 
 Fix $r \le n_x$ and assume $\lambda_r^{(i)}, \lambda_r^{(j)} >0$. Define $\mathbf{S}[i,j]$ according to \eqref{eq:avg_s},
with $s_{ij}$ computed using the top-$r$ eigenvectors of system $j$.
If $\mathbf{S}[i,j] =1 $, then for all $q \in [r]$, $\mathbf{v}_q^{(j)} = \pm \mathbf{v}_q^{(i)}$. 
\end{lemma}

\begin{proof}
    Assume $\mathbf{S}[i,j] = 1$. This implies that $s_{ij}=s_{ji}=1$. Since each factor $\rho_q^{(i,\,j)} \in [0,1]$, we have $\rho_{q}^{(i,\,j)}= 1 $ for all $q \in [r]$, which implies $ a_{q}^{(i,j)} = \lambda_q^{(i)} = \left\| \mathbf{\Sigma}^{(i)} \, \mathbf{v}_{q}^{(j)}\right\|_2$. 
    
    We proceed by induction on $q$. For $q=1$, since $\left\|\mathbf{v}_1^{(j)}\right\|_2 =1$, we have $\left\|\mathbf{\Sigma}^{(i)} \mathbf{v}_1^{(j)}\right\|_2 \leq \left\|\mathbf{\Sigma}^{(i)}\right\|_2 = \lambda_1^{(i)}$, with equality holding by assumption. Since $\lambda_1^{(i)}$ is simple, its
associated eigenspace is one-dimensional. Hence, equality holds if and only if
$\mathbf{v}_1^{(j)} = \pm \mathbf{v}_1^{(i)}$. This proves the base case.  

Now assume $\mathbf{v}_{\ell}^{(j)} = \pm \mathbf{v}_{\ell}^{(i)}$ for all $\ell < q$. Then $\mathbf{v}_q^{(j)}$ is orthogonal to $\mathrm{span}(\mathbf{v}_1^{(j)}, \dots,\mathbf{v}_{q-1}^{(j)})$, which equals $\mathrm{span}(\mathbf{v}_1^{(i)}, \dots,\mathbf{v}_{q-1}^{(i)})$. By the variational characterization of eigenvalues,
\begin{equation}
    \lambda_q^{(i)} = \max_{\substack{\|\mathbf{x}\|_2=1 \\  \mathbf{x} \perp \mathrm{span}(\mathbf{v}_1^{(i)}, \dots,\mathbf{v}_{q-1}^{(i)})}}\|\mathbf{\Sigma}^{(i)} \mathbf{x}\|_2.
\end{equation}
Since $\lambda_q^{(i)}  = \|\mathbf{\Sigma}^{(i)} \mathbf{v}_q^{(j)}\|_2$, $\mathbf{v}_q^{(j)}$ attains the maximum. Since $\lambda_q^{(i)}$ is simple, the maximizer is unique up to a sign change, and therefore $\mathbf{v}_q^{(j)} = \pm \mathbf{v}_q^{(i)}$. This concludes the proof.
\end{proof}

\begin{remark}
Under the same setting, the reverse implication also holds: if $\mathbf{v}_q^{(j)} = \pm \mathbf{v}_q^{(i)}$ for all $q \in [r]$, then $s_{ij} = 1$. This follows directly from the definition of $a_q^{(i,j)}$ and substitution into \eqref{eq:lambd(i,j)} and \eqref{eq:s(i,j)}. Similarly, $s_{ji} = 1$, and therefore $\mathbf{S}[i,j] = 1$.
\end{remark}

This interpretation of the similarity score motivates the following theoretical analysis of finite-sample effects on eigenspace alignment, i.e., the agreement between the leading covariance eigenspaces of different systems. Stronger alignment indicates that the systems distribute state variation along similar dominant directions. Specifically, we address two questions. First, we study how accurately the population covariance matrices can be estimated from finitely many trajectories, as characterized by Lemma~\ref{lemma:cov_est}, and how the sampling time $T$ influences this estimation through the system dynamics, as shown in Lemma~\ref{lemma:dynamics_dependency}. 
Second, we analyze how covariance estimation error propagates into deviations between eigenspaces using the Davis--Kahan theorem~\cite{yu2015useful}. These steps provide the foundation for bounding false eigenspace alignment induced by finite-sample covariance estimation.

%%%%%%%%%%%%%%%%%%
\subsection{Covariance-Estimation Bound}\label{sub:cov_est_bnd}
We begin by establishing a finite-sample concentration bound for the covariance estimation error in the spectral norm, $\| \widehat{\mathbf{\Sigma}}^{(i)} - \mathbf{\Sigma}^{(i)}\|_2$. The following bound is adapted from~\cite{Vershynin_2026}.

\begin{lemma}[Covariance Estimation Error]\label{lemma:cov_est}
    Consider the $i$-th system, and suppose that $\mathbf{x}_{1,T}^{(i)}, \dots, \mathbf{x}_{N_i,T}^{(i)}$ are i.i.d. random vectors in $\mathbb{R}^{n_x}$ distributed as $\mathcal{N}(\mathbf{0}, \mathbf{\Sigma}^{(i)})$. Then, for all $\xi > 0$,
\begin{align}\label{eq:cov_est}
&\mathbb{P}\Bigl( \|\widehat{\mathbf{\Sigma}}^{(i)} - \mathbf{\Sigma}^{(i)}\|_2 \ge \xi \Bigr) \nonumber \\
& \qquad \le 2 \,  \exp \left(c' n_x - C N_i \min \left\{ \frac{\xi^2}{\|\mathbf{\Sigma}^{(i)}\|_2^2}, \frac{\xi}{\|\mathbf{\Sigma}^{(i)}\|_2}\right\} \right), 
\end{align}
where $c'=\log 9$ and $C>0$ is a universal constant.
\end{lemma}
\begin{proof}
    The result follows from standard non-asymptotic covariance concentration inequalities for sub-Gaussian random vectors~\cite{Vershynin_2026}. Since $\mathbf{x}_{j,T}^{(i)} \sim \mathcal{N}(\mathbf{0}, \mathbf{\Sigma}^{(i)})$, for every unit vector $\mathbf{v} \in \mathbb{R}^{n_x}$, the scalar projection $\mathbf{v}^{\top} \mathbf{x}_{j,T}^{(i)}$ is Gaussian with variance $\mathbf{v}^\top \mathbf{\Sigma}^{(i)} \mathbf{v} \leq \|\mathbf{\Sigma}^{(i)}\|_2$. Thus, the sub-Gaussian variance proxy is bounded by $\|\mathbf{\Sigma}^{(i)}\|_2$. Applying the standard covariance concentration bound with dimension $n_x$, sample size $N_i$, and deviation level $\xi$ yields \eqref{eq:cov_est}. The factor $9^{n_x}$ arises from a $1/4$-net argument; see, e.g.,~\cite{fine}.
\end{proof}

To characterize the effect of the system parameters and sampling horizon $T$ on covariance estimation, we first express the terminal state of system $i$ along trajectory $l$ as
\begin{equation}
    \mathbf{x}_{l,T}^{(i)} =  (\mathbf{A}^{(i)})^T \mathbf{x}_{l,0}^{(i)} + \mathbf{G}_T^{(i)} \tilde{\mathbf{u}}^{(i)}_{l} + \mathbf{F}_T^{(i)} \tilde{\mathbf{w}}^{(i)}_{l} 
    = ~\mathbf{H}_T^{(i)} ~\boldsymbol{\eta}_{l}^{(i)}.\label{eq:x_t}
\end{equation}
Here,
\begin{align}
\mathbf{G}_T^{(i)}
&=
\big[
(\mathbf{A}^{(i)})^{T-1}\mathbf{B}^{(i)}
\;\;
(\mathbf{A}^{(i)})^{T-2}\mathbf{B}^{(i)}
\;\;
\cdots
\;\;
\mathbf{B}^{(i)}
\big], \nonumber\\
\mathbf{F}_T^{(i)}
&=
\big[
(\mathbf{A}^{(i)})^{T-1}
\;\;
(\mathbf{A}^{(i)})^{T-2}
\;\;
\cdots
\;\;
\mathbf{I}
\big], \nonumber\\
\mathbf{H}_T^{(i)}
&=
\big[
(\mathbf{A}^{(i)})^T
\;\;
\mathbf{G}_T^{(i)}
\;\;
\mathbf{F}_T^{(i)}
\big].
\label{eq:G_F_H_defs}
\end{align}
Moreover, with $\operatorname{col}(\cdot)$ denoting vertical concatenation,
\[
\tilde{\mathbf{u}}_l^{(i)}
=
\operatorname{col}(\mathbf{u}_{l,0}^{(i)},\ldots,\mathbf{u}_{l,T-1}^{(i)}),
\quad
\tilde{\mathbf{w}}_l^{(i)}
=
\operatorname{col}(\mathbf{w}_{l,0}^{(i)},\ldots,\mathbf{w}_{l,T-1}^{(i)}),
\]
and $
\boldsymbol{\eta}_l^{(i)}
=
\operatorname{col}
\left(
\mathbf{x}_{l,0}^{(i)},
\tilde{\mathbf{u}}_l^{(i)},
\tilde{\mathbf{w}}_l^{(i)}
\right).
$

\begin{lemma}[System Dynamics Dependency]\label{lemma:dynamics_dependency}
Given the representation $\mathbf{x}_{l,T}^{(i)}=\mathbf{H}_T^{(i)}\boldsymbol{\eta}_{l}^{(i)}$, the population covariance
$\mathbf{\Sigma}^{(i)} := \mathbb{E}[
\mathbf{x}_{l,T}^{(i)}\mathbf{x}_{l,T}^{(i)\top}]$
satisfies $ \mathbf{\Sigma}^{(i)} = \mathbf{H}_T^{(i)}
\mathbf{\Sigma}_{\eta}^{(i)} \mathbf{H}_T^{(i)\top}$, where $ \mathbf{\Sigma}_{\eta}^{(i)} := \mathbb{E}[ \boldsymbol{\eta}_{l}^{(i)} \boldsymbol{\eta}_{l}^{(i)\top}]$.
Consequently,
\begin{equation}\label{eq:final_sigma_bound}
\|\mathbf{\Sigma}^{(i)}\|_2 
\le \bar{\sigma}_i^2 \left(
\|\mathbf{A}^{(i)}\|_2^{2T} 
+ \left(1 + \|\mathbf{B}^{(i)}\|_2^2\right) 
\sum_{k=0}^{T-1} \|\mathbf{A}^{(i)}\|_2^{2k}
\right),
\end{equation}
where $\bar{\sigma}_i^2 := \left\|\mathbf{\Sigma}_{\eta}^{(i)}\right\|_2=\max\left\{ \sigma_{x,i}^2,\; \sigma_{u,i}^2,\; \sigma_{w,i}^2 \right\}$.
\end{lemma}

\begin{proof} 
From the covariance decomposition in the lemma, we have 
\begin{align}\label{eq:s_bound}
\left\|\mathbf{\Sigma}^{(i)}\right\|_2 =  \left\|\mathbf{H}_T^{(i)}
\mathbf{\Sigma}_{\eta}^{(i)} \mathbf{H}_T^{(i)\top}\right\|_2 \leq \left\|\mathbf{H}_T^{(i)} \right\|_2^2 \left\|
\mathbf{\Sigma}_{\eta}^{(i)}\right\|_2.
\end{align}
    By the block-matrix norm bound in~\cite[Proposition~3]{shamrai2025perturbation}, we have 
    \begin{align}  \label{eq:h_bound}
        \left\|\mathbf{H}_T^{(i)} \right\|_2^2 \leq \left\|\mathbf{A}^{(i)} \right\|_2^{2T} + \left\|\mathbf{G}_T^{(i)} \right\|_2^{2} + \left\|\mathbf{F}_T^{(i)} \right\|_2^{2}.
    \end{align}
    Moreover, one can write
        \begin{align} \label{eq:g_bound}
        &\left\|\mathbf{G}_T^{(i)} \right\|_2^2 \leq \left\|\mathbf{B}^{(i)}\right\|_2^{2}\, \sum_{k=0}^{T-1} \left\|\mathbf{A}^{(i)}\right\|_2^{2k}, \quad
        \left\|\mathbf{F}_T^{(i)} \right\|_2^2 \leq  \sum_{k=0}^{T-1} \left\|\mathbf{A}^{(i)}\right\|_2^{2k}.
    \end{align}
    Since $\mathbf{\Sigma}_{\eta}^{(i)}$ is a block diagonal matrix, $\left\|\mathbf{\Sigma}_{\eta}^{(i)}\right\|_2 = \max\left\{ \sigma_{x,i}^2,\, \sigma_{u,i}^2,\, \sigma_{w,i}^2 \right\}$. Substituting \eqref{eq:g_bound} into \eqref{eq:h_bound} and combining the resulting bound with \eqref{eq:s_bound} yields \eqref{eq:final_sigma_bound}.
\end{proof}

\begin{remark}[Fundamental trade-offs in sample complexity]
Lemma~\ref{lemma:dynamics_dependency} shows how the sample complexity depends on the system matrices and sampling horizon. When $\|\mathbf{A}^{(i)}\|_2<1$, the covariance norm remains uniformly bounded as $T\to\infty$, namely
$\|\mathbf{\Sigma}^{(i)}\|_2 \leq \bar{\sigma}_i^2(1+\|\mathbf{B}^{(i)}\|_2^2)/(1-\|\mathbf{A}^{(i)}\|_2^2)$. For a fixed accuracy $\xi$, it suffices to use $N_i =\mathcal{O}(n_x)$ trajectories, independently of $T$. When $\|\mathbf{A}^{(i)}\|_2=1$, the covariance norm is bounded above by a quantity growing linearly with $T$, yielding $N_i=\mathcal{O}(n_xT^2)$ to maintain fixed accuracy. When $\|\mathbf{A}^{(i)}\|_2>1$, the covariance norm is bounded above by a quantity growing exponentially with $T$, yielding the sufficient sample size $N_i=\mathcal{O}(n_x\|\mathbf{A}^{(i)}\|_2^{4T})$.
\end{remark}

The above bounds quantify the accuracy of covariance estimation. We next show how this estimation error translates into eigenspace misalignment through the Davis--Kahan theorem.

\subsection{Eigenspace Perturbation Bound}
%%%%
For system $i$, we decompose the eigenspaces as
\[
     \mathbf{V}^{(i)}= [\mathbf{V}_0^{(i)} \; \mathbf{V}_1^{(i)}],
     \qquad
     \widehat{\mathbf{V}}^{(i)}
     = [\widehat{\mathbf{V}}_0^{(i)} \; \widehat{\mathbf{V}}_1^{(i)}],
\]
where $\mathbf{V}_0^{(i)}$ and $\widehat{\mathbf{V}}_0^{(i)}$ span the top-$r$ eigenspaces of
$\mathbf{\Sigma}^{(i)}$ and $\widehat{\mathbf{\Sigma}}^{(i)}$, respectively, and
$\mathbf{V}_1^{(i)}$ and $\widehat{\mathbf{V}}_1^{(i)}$ span their orthogonal complements.
Let 
\[
    \mathbf{P}^{(i)} = \mathbf{V}_0^{(i)}\mathbf{V}_0^{(i)\top}, \quad 
    \widehat{\mathbf{P}}^{(i)}
    = \widehat{\mathbf{V}}_0^{(i)}\widehat{\mathbf{V}}_0^{(i)\top}
\]
denote the corresponding orthogonal projection matrices.

\begin{corollary}[Eigenspace Perturbation for System $i$]
\label{cor:pert_bound} 
Let $1\le r<n_x$, and suppose that the population covariance matrix $\mathbf{\Sigma}^{(i)}$ has a positive eigengap, and define $\delta^{(i)} = \frac{\lambda_r^{(i)} - \lambda_{r+1}^{(i)} }{2}>0$.
Then

\begin{equation}
    \left\|\mathbf{P}^{(i)}-\widehat{\mathbf{P}}^{(i)}\right\|_F
    \le
    \frac{\sqrt{2r}}{\delta^{(i)}}
    \left\|\widehat{\mathbf{\Sigma}}^{(i)}-\mathbf{\Sigma}^{(i)}\right\|_2 .
\end{equation}
\end{corollary}

\begin{proof}
For two $r$-dimensional subspaces,
$
\|\mathbf{P}^{(i)}-\widehat{\mathbf{P}}^{(i)}\|_F
=
\sqrt{2}\|(\widehat{\mathbf{V}}_1^{(i)})^\top\mathbf{V}_0^{(i)}\|_F.
$
By the population-eigengap variant of Davis--Kahan $\sin\Theta$ theorem~\cite{yu2015useful},
\[
\|(\widehat{\mathbf{V}}_1^{(i)})^\top\mathbf{V}_0^{(i)}\|_F
\le
\frac{ \sqrt{r}}{\delta^{(i)}}
\|
\widehat{\mathbf{\Sigma}}^{(i)}-\mathbf{\Sigma}^{(i)}\|_2.
\]
Combining the two inequalities proves the result. 

\end{proof}

Combining Lemma~\ref{lemma:cov_est} with Corollary~\ref{cor:pert_bound} gives the following finite-sample bound on eigenspace misalignment.

\begin{theorem}[Finite-Sample Eigenspace Misalignment Bound]\label{thm:eigenspace_misalign}
Under the assumptions of Lemma~\ref{lemma:cov_est} and Corollary~\ref{cor:pert_bound}, let $\mathcal{E}_i(\tau) =\left\{\left\|\mathbf{P}^{(i)} - \widehat{\mathbf{P}}^{(i)} \right\|_F  \ge \tau \right\}$ denote the eigenspace misalignment event for system $i$. Then, for all $\tau>0$, 
\begin{equation}\label{ineq:mis_align}
    \begin{aligned}
        \mathbb{P}\left(\mathcal{E}_i(\tau)\right) 
        &\le 2\exp\Bigg(c'n_x - C N_i \min \Bigg\{ \frac{\tau^2(\delta^{(i)})^2  }{2r\|\mathbf{\Sigma}^{(i)}\|_2^2}, 
        \\
        & \hspace{4.3cm}
        \frac{\tau \delta^{(i)}}{\sqrt{2r}\|\mathbf{\Sigma}^{(i)}\|_2}\Bigg\} \Bigg),
    \end{aligned}
\end{equation}
    where $c'=\log 9$ and $C>0$ is as in Lemma~\ref{lemma:cov_est}.
\end{theorem}

\begin{proof}
By Corollary~\ref{cor:pert_bound},
\[
    \mathcal{E}_i(\tau)
    \subseteq
    \left\{
    \left\|
    \widehat{\mathbf{\Sigma}}^{(i)}-\mathbf{\Sigma}^{(i)}
    \right\|_2
    \ge
    \frac{\tau\delta^{(i)}}{\sqrt{2r}}
    \right\}.
\]
Applying Lemma~\ref{lemma:cov_est} with
$\xi=\tau\delta^{(i)}/\sqrt{2r}$ yields \eqref{ineq:mis_align}.
\end{proof}
\subsection{Finite-Sample Inter-Cluster Eigenspace Separation}
The finite-sample eigenspace misalignment bound developed above allows us to quantify the probability that systems from different true clusters have nearly aligned estimated eigenspaces despite their population separation. We next derive a pairwise probability bound and extend it to a uniform guarantee over all cross-cluster pairs. Let 
\begin{equation}
d_{\min} = \min_{k \neq k'} \min_{\substack{i \in \mathcal{C}_k \\ j \in \mathcal{C}_{k'}}} \left\| \mathbf{P}^{(i)} - \mathbf{P}^{(j)} \right\|_F
\end{equation}
denote the minimum inter-cluster separation between the true eigenspaces. For an alignment threshold $0 < \gamma < d_{\min}$, we define the pairwise false eigenspace-alignment event for any pair of systems $i$ and $j$ from distinct clusters, i.e., $i \in \mathcal{C}_k$ and $j \in \mathcal{C}_{k'}$ with $k \neq k'$, as

\begin{equation}\label{eq:false_merge_event}
    \Xi_{ij}(\gamma) = \left\{ \left\|\widehat{\mathbf{P}}^{(i)} - \widehat{\mathbf{P}}^{(j)}\right\|_F \leq \gamma \right\}.
\end{equation}

\begin{theorem}[Probability of Pairwise False Eigenspace Alignment]
\label{thm:pairwise_false_merge}
Let $i \in \mathcal{C}_k$ and $j \in \mathcal{C}_{k'}$ with $k \neq k'$. Suppose
the assumptions of Theorem~\ref{thm:eigenspace_misalign} hold for both systems $i$ and $j$ and $0<\gamma<d_{\min}$. Then
\begin{equation}
    \mathbb{P}\left(\Xi_{ij}(\gamma)\right)
    \le
    2 \exp \left(c'n_x-\mathcal{K}_{i}(\gamma)\right)
    +
    2 \exp \left(c'n_x-\mathcal{K}_{j}(\gamma)\right),
\end{equation}
where, for $m\in\{i,j\}$,
\begin{equation}
    \mathcal{K}_{m}(\gamma)
    :=
    C N_m
    \min \left\{
    \frac{(d_{\min}-\gamma)^2(\delta^{(m)})^2}
    {8r\|\mathbf{\Sigma}^{(m)}\|_2^2},
    \frac{(d_{\min}-\gamma)\delta^{(m)}}
    {\sqrt{8r}\|\mathbf{\Sigma}^{(m)}\|_2}
    \right\}.
\end{equation}
\end{theorem}

\begin{proof}
Since $i$ and $j$ belong to different clusters,
$\|\mathbf{P}^{(i)}-\mathbf{P}^{(j)}\|_F\ge d_{\min}$. Therefore, by the triangle inequality,
\[
    \|\widehat{\mathbf{P}}^{(i)}-\widehat{\mathbf{P}}^{(j)}\|_F
    \ge
    d_{\min}
    -
    \|\mathbf{P}^{(i)}-\widehat{\mathbf{P}}^{(i)}\|_F
    -
    \|\mathbf{P}^{(j)}-\widehat{\mathbf{P}}^{(j)}\|_F .
\]
Thus, on the false eigenspace-alignment event $\Xi_{ij}(\gamma)$,
\[
    \|\mathbf{P}^{(i)}-\widehat{\mathbf{P}}^{(i)}\|_F
    +
    \|\mathbf{P}^{(j)}-\widehat{\mathbf{P}}^{(j)}\|_F
    \ge
    d_{\min}-\gamma .
\]
Therefore, by the union bound,
\[
    \mathbb{P}\left(\Xi_{ij}(\gamma)\right)
    \le
    \mathbb{P}\left(\mathcal{E}_i\left(\frac{d_{\min}-\gamma}{2}\right)\right)
    +
    \mathbb{P}\left(\mathcal{E}_j\left(\frac{d_{\min}-\gamma}{2}\right)\right).
\]
Applying Theorem~\ref{thm:eigenspace_misalign} with
$\tau=(d_{\min}-\gamma)/2$ gives the stated bound.
\end{proof}
\begin{remark}[Interaction between separation and eigengap]
The dependence of $\mathcal{K}_{m}(\gamma)$ on the product $\delta^{(m)}(d_{\min} -\gamma)$ highlights the interaction between the inter-cluster separation and the spectral gap. 
A larger margin $d_{\min} - \gamma$ makes false eigenspace alignment less likely, while a larger eigengap $\delta^{(m)}$ makes the leading eigenspace more stable under covariance estimation error. Conversely, a small eigengap can limit the benefit of a large inter-cluster separation, since the estimated eigenspace becomes more sensitive to finite-sample perturbations.  
\end{remark}

\begin{remark}[Dominant system in the pairwise bound]
The pairwise false eigenspace-alignment probability bound is controlled by both $\mathcal{K}_i(\gamma)$ and $\mathcal{K}_j(\gamma)$. Hence, the bound is dominated by the system with fewer trajectories, a smaller eigengap, or a larger covariance norm. This reflects that preserving inter-cluster eigenspace separation requires accurate eigenspace estimation for both systems.    
\end{remark}

\begin{corollary}[Uniform Inter-Cluster Eigenspace Separation]
Suppose that the assumptions of Theorem~\ref{thm:pairwise_false_merge} hold for every pair of systems from different true clusters. Define the global false eigenspace-alignment event 
    \begin{equation*}
          \Xi_{\text{FM}}(\gamma) := \bigcup_{1\leq k < k' \leq K} \bigcup_{i\in \mathcal{C}_k}  \bigcup_{j\in \mathcal{C}_{k'}} \Xi_{ij}(\gamma),
    \end{equation*}
    which is the event that at least one pair of systems from distinct true clusters has estimated projector distance at most $\gamma$. Then
    \begin{align}
        \mathbb{P} \left(\Xi^c_{\text{FM}}(\gamma)\right) \geq 1 - 2 M^2 \exp{ (c'n_x - \mathcal{K}_{\min}(\gamma))},
    \end{align}
where $\mathcal{K}_{\min}(\gamma) := \min_{m \in [M]} \mathcal{K}_{m}(\gamma)$.
\end{corollary}
\begin{proof}
    By the union bound, 
    \begin{equation*}
                 \mathbb{P} \left( \Xi_{\text{FM}}(\gamma) \right) \le  \sum_{1\leq k < k' \leq K} \sum_{i\in \mathcal{C}_k}  \sum_{j\in \mathcal{C}_{k'}} \mathbb{P} \left(\Xi_{ij}(\gamma)\right).
    \end{equation*}
   Applying Theorem~\ref{thm:pairwise_false_merge}, and using $\mathcal{K}_{\min}(\gamma) = \min_{m \in [M]} \mathcal{K}_{m}(\gamma)$, we obtain 
    \begin{equation*}
          \mathbb{P} \left( \Xi_{\text{FM}}(\gamma) \right) \le 4 \sum_{1\leq k < k' \leq K} \left|\mathcal{C}_k\right|  \left|\mathcal{C}_{k'}\right|\exp \left(c'n_x - \mathcal{K}_{\min}(\gamma)\right).
    \end{equation*}
    Since $ \sum_{1\leq k < k' \leq K} \left|\mathcal{C}_k\right|  \left|\mathcal{C}_{k'}\right| \le \frac{M^2}{2}$, it follows that 
        \begin{equation*}
          \mathbb{P} \left( \Xi_{\text{FM}}(\gamma) \right) \le 2 M^2 \exp{ (c'n_x - \mathcal{K}_{\min}(\gamma))}.
    \end{equation*}
    Therefore, $\mathbb{P} \left(\Xi_{\text{FM}}^c(\gamma)\right) \geq 1 - 2 M^2 \exp \left(c'n_x - \mathcal{K}_{\min}(\gamma)\right).$
\end{proof}

\begin{remark}[Sample Complexity for Uniform Eigenspace Separation]
To attain uniform inter-cluster eigenspace separation with probability at least $1-\epsilon$, where $\epsilon \in (0,1)$, it is sufficient that the number of trajectories for each system $m \in [M]$ satisfies
\begin{align*}
    N_m \ge &\frac{1}{C}\left({c'} n_x + \log \left(\frac{2}{\epsilon}\right)+2\log  M \right) \times \\
    &\max \left\{ \frac{8r \left\|\mathbf{\Sigma}^{(m)}\right\|_2^2}{(d_{\min} - \gamma)^2 (\delta^{(m)})^2}, \frac{\sqrt{8r} \left\|\mathbf{\Sigma}^{(m)}\right\|_2}{(d_{\min} - \gamma)\delta^{(m)}}\right\}.
\end{align*}
    Thus, the required number of trajectories grows linearly with the state dimension, logarithmically with the number of systems, and increases as the inter-cluster separation margin $d_{\min} - \gamma$ or the eigengap $\delta^{(m)}$ becomes smaller.
\end{remark}

\section{Experiments}\label{sec:experiments}

%\subsection{Numerical Results}
\begin{figure*}[htbp]
%\hfill
\centering
%\vspace{1.5cm}
\subfigure[$k= 1$.]{\includegraphics[width=0.32\textwidth]{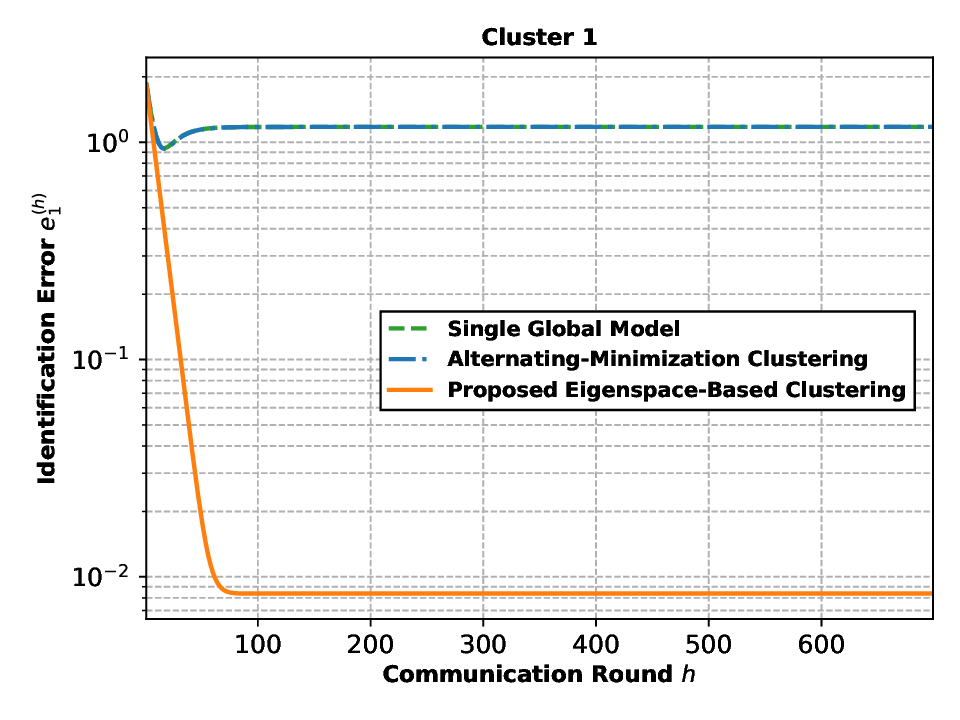}}
% \hfill
%\vspace{1.5cm}
\subfigure[ $k = 2$.]{\includegraphics[width=0.32\textwidth]{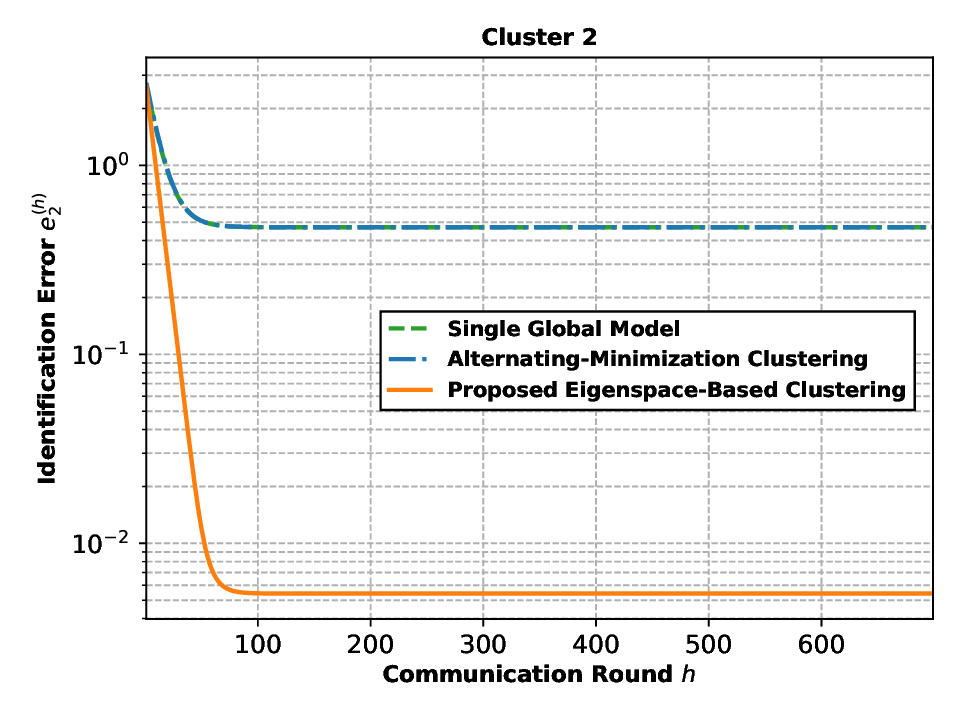}}
% \hfill
%\vspace{-3.5cm}
\subfigure[ $k = 3$.]{\includegraphics[width=0.32\textwidth]{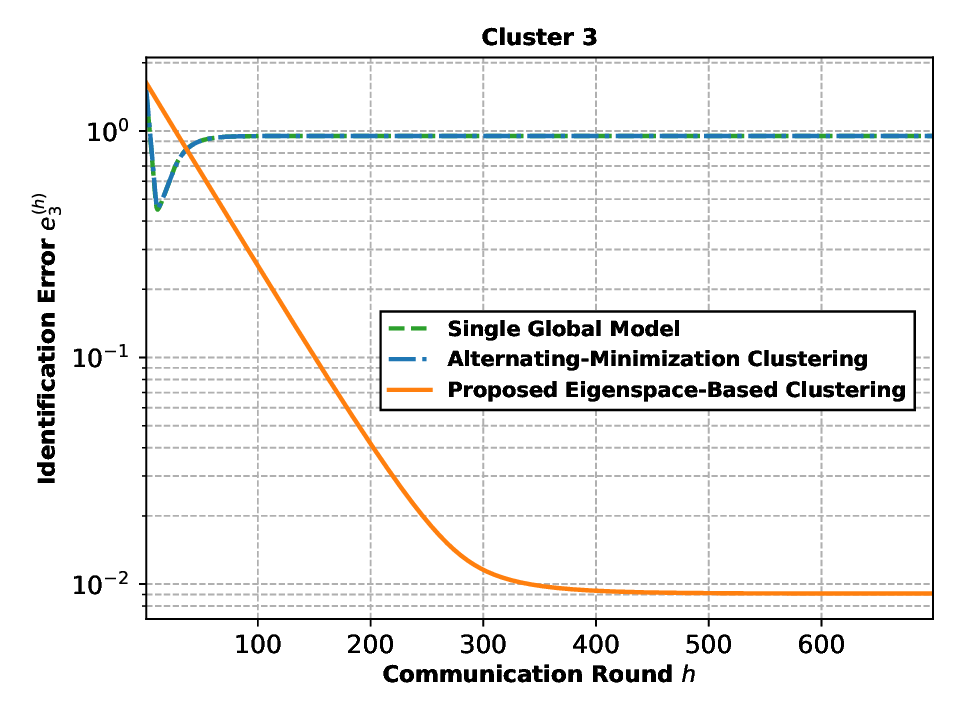}}
% \hfill

\caption{Cluster-wise identification error $e_k^{(h)}$ under random initialization for the three true clusters.}
\label{fig:cluster_wise_sysid_baselines}
\vspace{-.1in}
\end{figure*}

%%%%%%%%%%
\begin{figure*}[htbp]
%\hfill
\centering
%\vspace{1.5cm}
\subfigure[$k=1$.]{\includegraphics[width=0.32\textwidth]{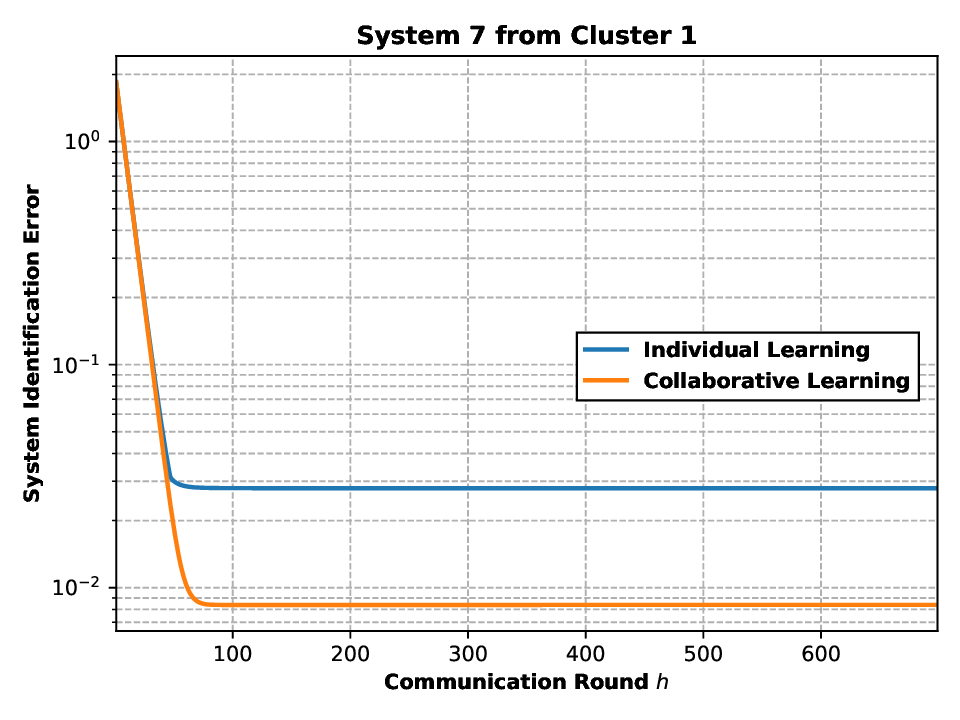}}
% \hfill
%\vspace{1.5cm}
\subfigure[$k=2$.]{\includegraphics[width=0.32\textwidth]{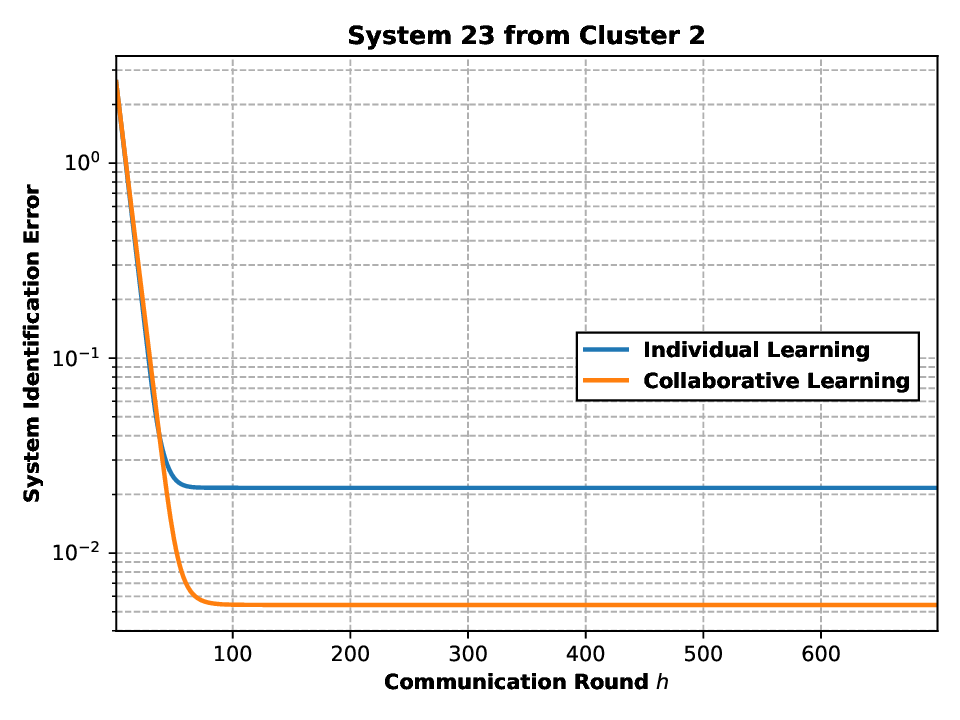}}
% \hfill
%\vspace{-3.5cm}
\subfigure[$k=3$.]{\includegraphics[width=0.32\textwidth]{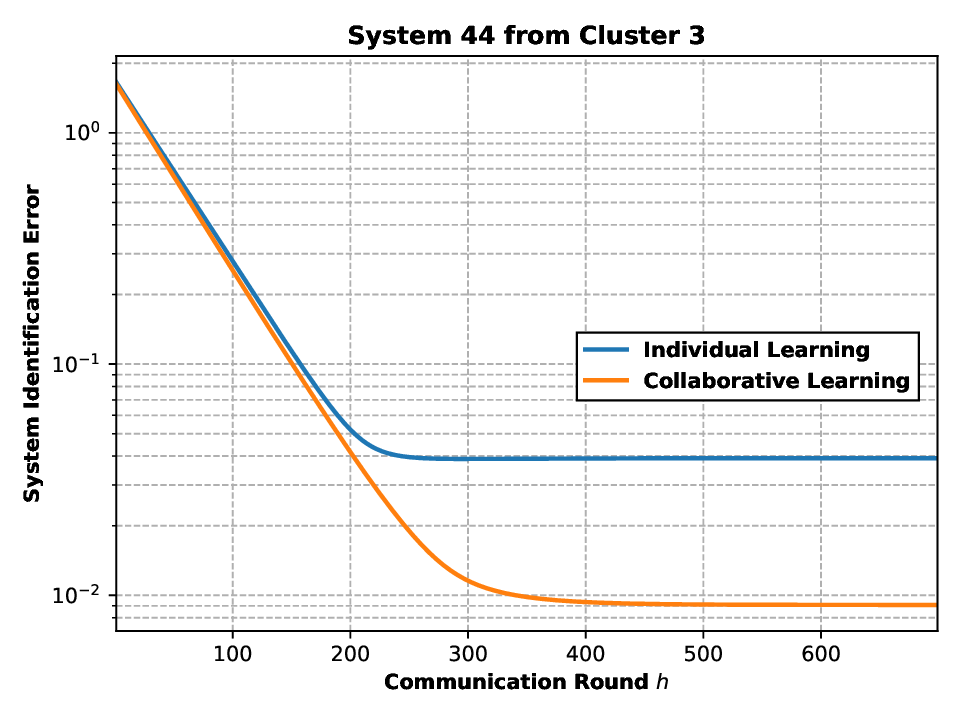}}
% \hfill
\caption{Benefit of collaborative learning.} 

\label{fig:fl_incentive}
\vspace{-.1in}
\end{figure*}

In this section, we evaluate the proposed eigenspace-based clustered system identification method on a set of heterogeneous LTI systems. The following simulations illustrate the performance of the proposed clustering method summarized in Algorithm~1.

\noindent \textbf{System dynamics}. The simulation setup is adapted from the numerical example in~\cite{toso2023learning}, which considers $M=50$ systems, each described by an LTI model as in~\eqref{eq:state-space}, with $K=3$ clusters. The number of systems in each cluster is given by $|\mathcal{C}_1|=10$, $|\mathcal{C}_2|=24$, and $|\mathcal{C}_3|=16$.
The cluster-specific system matrices $(\mathbf{A}_k,\mathbf{B}_k)$ are chosen as in~\cite{toso2023learning}.
The standard deviations of the initial state, control input, and process noise are set uniformly within each cluster as follows:
$\sigma_{x,i} = \sigma_{u,i} = \sigma_{w,i} = 0.11$ for all $i \in \mathcal{C}_1$,
$\sigma_{x,i} = \sigma_{u,i} = \sigma_{w,i} = 0.12$ for all $i \in \mathcal{C}_2$, and
$\sigma_{x,i} = \sigma_{u,i} = \sigma_{w,i} = 0.05$ for all $i \in \mathcal{C}_3$.

We consider the same number of trajectories for all systems, with $N_i = 150$ for all $i \in [M]$, and set the trajectory length to $T = 50$. A fixed step size $\beta = 10^{-3}$ is used for all clusters. For the proposed eigenspace-based clustering step, we use the top $r=2$ eigenvectors of each local sample covariance matrix.

\noindent \textbf{Performance evaluation}. At each iteration $h$, for each system $i$, we define the assigned-model error as $
    \varepsilon_{i}^{(h)} = \left\|\widehat{\boldsymbol{\Theta}}_{\widehat{c}(i)}^{(h)} - \boldsymbol{\Theta}_{c(i)}\right\|_2, $
where $\widehat{c}(i)$ denotes the estimated cluster identity assigned to system $i$, and $c(i)$ denotes its true cluster identity. The cluster-wise error for true cluster $k$ is then defined as $
    e_{k}^{(h)} = \frac{1}{\left|\mathcal{C}_k\right|}\sum_{i\in \mathcal{C}_k } \varepsilon_{i}^{(h)}$.

\noindent \textbf{Baselines}. 
We compare the proposed one-shot eigenspace-based clustering algorithm with two baselines. The first is the alternating-minimization method of~\cite{toso2023learning}. The second is a non-clustered approach in which system heterogeneity is ignored and a single global model is learned from all systems. For a fair comparison, all methods are initialized randomly.

Fig.~\ref{fig:cluster_wise_sysid_baselines} demonstrates the effectiveness of the proposed eigenspace-based clustering method compared to the baselines. The proposed method consistently achieves lower identification error across the three clusters. In contrast, the alternating-minimization baseline converges to a performance level nearly identical to that of the single global model. This overlap occurs because, under random initialization, the alternating-minimization procedure does not reliably recover the true cluster structure and tends to merge the systems into a single cluster. This observation is consistent with the initialization requirement in~\cite{toso2023learning}, where the initial cluster models are assumed to lie sufficiently close to the true system models, and highlights the sensitivity of training-dependent cluster assignment to initialization. Meanwhile, the single global model illustrates the importance of accounting for system heterogeneity through clustering.

Fig.~\ref{fig:fl_incentive} illustrates the learning gains obtained through collaboration. Individual learning relies only on the trajectories collected by a single system, whereas collaborative learning leverages information from multiple systems within the same estimated cluster. This additional information improves model estimation and leads to lower identification error.
%%%%%%%%%

\begin{figure}
    \centering
    \includegraphics[width=0.90\linewidth]{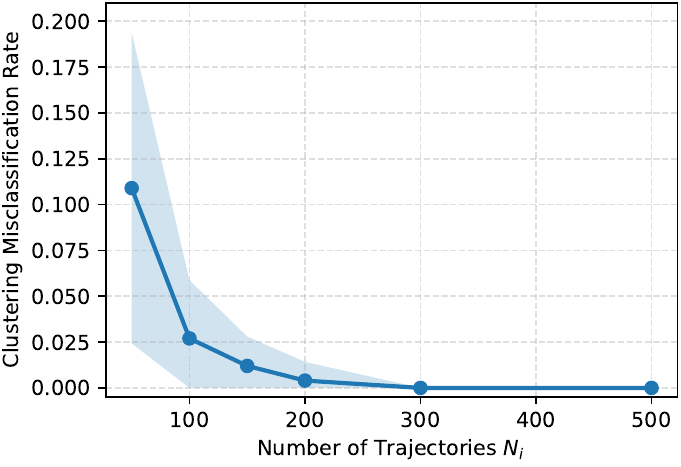}
    \caption{Effect of the number of trajectories on clustering performance. The curve reports the mean clustering misclassification rate over 20 independent trials, and the shaded region indicates one standard deviation.}  
\label{fig:misclass_rate} \vspace{-4mm}
\end{figure}

Fig.~\ref{fig:misclass_rate} shows that the clustering misclassification rate decreases as the number of trajectories $N_i$ increases. This trend is consistent with the finite-sample eigenspace perturbation analysis, since larger $N_i$ improves covariance estimation and leads to more reliable eigenspace-based clustering.

\section{Conclusion}\label{sec:conclusion}
We proposed a communication-efficient and training-free clustering algorithm for personalized federated system identification under heterogeneous system dynamics. The proposed method infers cluster identities by quantifying the alignment between the leading covariance eigenspaces associated with different systems. We provided a finite-sample theoretical analysis that characterizes eigenspace perturbations induced by covariance estimation error, bounds the probability of pairwise false eigenspace alignment, and establishes sufficient conditions for uniform inter-cluster eigenspace separation with high probability in terms of system-dependent quantities. Numerical results demonstrated that the proposed one-shot eigenspace-based clustering approach accurately identifies systems with shared dynamics and improves personalized model-estimation performance compared with iterative training-based clustering and non-clustered baselines.

Future work includes extending the framework to partially observable systems, incorporating additional privacy-preserving mechanisms for shared eigenspace information to limit potential information leakage, and studying hybrid schemes in which the proposed one-shot clustering step is used to initialize subsequent iterative cluster-refinement methods. 

\ifCLASSOPTIONcaptionsoff
  \newpage
\fi

%===
\bibliographystyle{unsrt}
\bibliography{Ali}

\end{document}